\documentclass[onecolumn,a4paper]{article}%
\usepackage{amsfonts}
\usepackage{amsmath}
\usepackage{amssymb}
\usepackage{graphicx}
\usepackage{hyperref}%
\newtheorem{theorem}{Theorem}[section]
\newtheorem{lemma}[theorem]{Lemma}

\newtheorem{remark}[theorem]{Remark}
\numberwithin{equation}{section}

\providecommand{\keywords}[1]
{
  \small	
  \textbf{\textit{Keywords---}} #1
}
\newenvironment{proof}[1][Proof]{\noindent\textbf{#1.} }{\ \rule{0.5em}{0.5em}}
\ifx\pdfoutput\relax\let\pdfoutput=\undefined\fi
\newcount\msipdfoutput
\ifx\pdfoutput\undefined\else
\ifcase\pdfoutput\else
\ifx\paperwidth\undefined\else
\ifdim\paperheight=0pt\relax\else\pdfpageheight\paperheight\fi
\ifdim\paperwidth=0pt\relax\else\pdfpagewidth\paperwidth\fi
\fi\fi\fi
\begin{document}
\title{Hierarchical Wilson-Cowan Models and Connection Matrices}
\author{W. A. Z\'{u}\~{n}iga-Galindo\thanks{The author was partially supported by the Lokenath Debnath Endowed Professorship.} \\ wilson.zunigagalindo@utrgv.edu \\
B. A. Zambrano-Luna\\
brian.zambrano@utrgv.edu@utrgv.edu
\and
University of Texas Rio Grande Valley\\
School of Mathematical \& Statistical Sciences\\
One West University Blvd\\
Brownsville, TX 78520, United States.}
\date{}
\maketitle

\begin{abstract}
This work aims to study the interplay between the Wilson-Cowan model and the
connection matrices. These matrices describe the cortical neural wiring, while
the Wilson-Cowan equations provide a dynamical description of neural
interaction. We formulate the Wilson-Cowan equations on locally compact
Abelian groups. We show that the Cauchy problem is well-posed. We then select
a type of group that allows us to incorporate the experimental information
provided by the connection matrices. We argue that the classical Wilson-Cowan
model is incompatible with the small-world property. A necessary condition to
have this property is that the Wilson-Cowan equations be formulated on a
compact group. We propose a $p$-adic version of the Wilson-Cowan model, a
hierarchical version in which the neurons are organized into an infinite
rooted tree. We present several numerical simulations showing that the
$p$-adic version matches the predictions of the classical version in relevant
experiments. The $p$-adic version allows the incorporation of the connection
matrices into the Wilson-Cowan model. We present several numerical simulations
using a neural network model that incorporates a $p$-adic approximation of the
connection matrix of the cat cortex.

\end{abstract}

\section{Introduction}

This work explores the interplay between the Wilson-Cowan models, the
connection matrices, and the non-Archimedean models of complex systems.

The Wilson-Cowan model describes the evolution of excitatory and inhibitory
activity in a synaptically coupled neuronal network. The model is given by the
following system of non-linear integro-differential evolution equations:%
\[%
\begin{array}
[c]{ll}%
\tau\frac{\partial E(x,t)}{\partial t}= & -E(x,t)+\\
& \\
& \left(  1-r_{E}E(x,t)\right)  S_{E}\left(  w_{EE}\left(  x\right)  \ast
E(x,t)-w_{EI}\left(  x\right)  \ast I(x,t)+h_{E}\left(  x,t\right)  \right) \\
& \\
\tau\frac{\partial I(x,t)}{\partial t}= & -I(x,t)+\\
& \\
& \left(  1-r_{I}I(x,t)\right)  S_{I}\left(  w_{IE}\left(  x\right)  \ast
E(x,t)-w_{II}(x)\ast I(x,t)+h_{I}\left(  x,t\right)  \right)  ,
\end{array}
\]
where $E(x,t)$ is a temporal coarse-grained variable describing the proportion
of excitatory neuron firing per unit of time at position $x\in\mathbb{R}$ at
the \ instant $t\in\mathbb{R}_{+}$. Similarly the variable $I(x,t)$ represents
the activity of the inhibitory population of neurons. The main parameters of
the model are the strength of the connections between each subtype of
population ($w_{EE}$, $w_{IE}$, $w_{EI}$, $w_{II}$) and the strength of input
to each subpopulation ($h_{E}\left(  x,t\right)  $, $h_{I}\left(  x,t\right)
$). This model generates a diversity of dynamical behaviors that are
representative of observed activity in the brain, like multistability,
oscillations, traveling waves, and spatial patterns, see, e.g.,
\cite{Wilson-Cowan 1}-\cite{Wilson-Cowan 2}, \cite{Neural-Fields} and the
references therein.

We formulate the Wilson-Cowan model on Abelian, locally compact topological
groups. The classical model corresponds to the group $(\mathbb{R},+)$. In this
framework, using classical techniques on semilinear evolution equations, see
e.g. \cite{C-H}-\cite{Milan}, we show that the corresponding Cauchy problem is
locally well-posed, and if $r_{E}=r_{I}=0$, it is globally well-posed, see
Theorem \ref{Theorem1}. This last condition corresponds to the case of two
coupled perceptrons.

Nowadays, there is a large amount of experimental data about the connection
matrices of the cerebral cortex of invertebrates and mammalians. Based on
these data, several researchers hypothesized that cortical neural networks are
arranged in fractal or self-similar patterns and have the small-world
property, see, e.g., \cite{Sporns et al 1}-\cite{Skoch et al}, and the
references therein. The connection matrices provide a static view of the
neural connections.

The investigation of the relationships between the Wilson-Cowan model and the
connection matrices is quite natural, since model was proposed to explain the
cortical dynamics, while the matrices describe the functional geometry of the
cortex. We initiate this study here.

A network having the small-world property necessarily has long-range
interactions; see Section \ref{Section_2}. In the Wilson-Cowan model, the
kernels ($w_{EE}$, $w_{IE}$, $w_{EI}$, $w_{II}$) describing the neural
interactions are Gaussian in nature, so only short-range interactions may
occur. For practical purposes, these kernels have compact support. On the
other hand, the Wilson-Cowan model on a general group requires that the
kernels be integrable; see Section \ref{Section_1}. We argue that $G$ must be
compact to satisfy the small-world property. Under this condition, any
continuous kernel is integrable. Wilson and Cowan formulated their model on
the group $(\mathbb{R},+)$. The only compact subgroup of this group is the
trivial one. The small-world property is, therefore, incompatible with the
classical Wilson-Cowan model.

It is worth noting that the absence of non-trivial compact subgroups in
$(\mathbb{R},+)$ is a consequence of the Archimedean axiom (the absolute value
is \ not bounded on the integers). Therefore, to avoid this problem, we can
replace $\mathbb{R}$ with a non-Archimedean field, which is a field where the
Archimedean axiom is no valid. We selected the field of the $p$-adic numbers.
This field has infinitely many compact subgroups, the balls with center at the
origin. We selected the unit ball, the ring of $p$-adic numbers $\mathbb{Z}%
_{p}$. The $p$-adic integers are organized in an infinite rooted tree. We use
this hierarchical structure as the topology for our $p$-adic version of the
Wilson-Cowan model. In principle, we could use other groups, such as the
classical compact groups, to replace $(\mathbb{R},+)$, but it is also
essential to have a rigorous study of the discretization of the model. For the
group $\mathbb{Z}_{p}$, this task can be performed using standard
approximation techniques for evolutionary equations, see, e.g., \cite[Section
5.4]{Milan}.

The $p$-adic Wilson-Cowan model admits good discretizations. Each
discretization corresponds to a system of non-linear integro-differential
equations on a finite rooted tree. We show that the solution of the Cauchy
problem of this discrete system provides a good approximation to the solution
of the Cauchy problem of the $p$-adic Wilson-Cowan model, see Theorem
\ref{Theorem2}.

We provide extensive numerical simulations of $p$-adic Wilson-Cowan models. In
Section \ref{Section_4}, we present three numerical simulations showing that
the $p$-adic models provide a similar explanation to the numerical experiments
presented in \cite{Wilson-Cowan 2}. In these experiments the kernels ($w_{EE}%
$, $w_{IE}$, $w_{EI}$, $w_{II}$) were chosen to have similar properties to the
kernels used in \cite{Wilson-Cowan 2}. In Section \ref{Section_5}, we consider
the problem of how to integrate the connection matrices into the $p$-adic
Wilson-Cowan model. This fundamental scientific task aims to use the vast
amount of data on maps of neural connections to understand the dynamics of the
cerebral cortex of invertebrates and mammalians. We show that the connection
matrix of the cat cortex can be well approximated by a $p$-adic kernel
$K_{r}(x,y)$. We then replace the excitatory-excitatory relation term
$w_{EE}\ast E$ by $\int_{\mathbb{Z}_{p}}K_{r}(x,y)E(y)dy$ but keep the other
kernels as in Simulation 1 presented in Section \ref{Section_4}. The response
of this network is entirely different from that given in Simulation 1. For the
same stimulus, the response of the last network exhibits very complex
patterns, while the response of the network presented in Simulation 1 is simpler.

The $p$-adic analysis has shown to be the right tool in the construction of a
wide variety of models of complex hierarchic systems, see, e.g.,
\cite{Aveetisov 1}-\cite{ZZ2}, and the references therein. Many of these
models involve abstract evolution equations of the type $\partial
_{t}u+\boldsymbol{A}u=F(u)$. In these models, the discretization of the
operator $\boldsymbol{A}$ is an ultrametric matrix $A_{l}=\left[
a_{ij}\right]  _{i,j\in G_{l}}$, where $G_{l}$ is a finite rooted tree with
$l$ levels and $p^{l}$ branches, here $p$ is a fixed prime number, see the
numerical simulations in \cite{ZZ1}-\cite{ZZ2}. Locally, the connection
matrices look very similar to the matrices $A_{l}$. The problem of
approximating large connection matrices by ultrametric matrices is an open problem.

\section{\label{Section_1}An abstract version of the Wilson-Cowan Equations}

In this section, we formulate the Wilson-Cowan model on locally compact
topological groups and study the well-posedness of the Cauchy problem attached
to these equations.

\subsection{Wilson-Cowan Equations on Abelian locally compact topological
groups}

Let $\left(  \mathcal{G},+\right)  $ be an Abelian, locally compact
topological group. Let $d\mu$ be a fixed Haar measure on $\left(
\mathcal{G},+\right)  $. The basic example is $\left(  \mathbb{R}%
^{N},+\right)  $, the $N$-dimensional Euclidean space\ considered as an
additive group. In this case $d\mu$ is the Lebesgue measure of $\mathbb{R}%
^{N}$.

Let $L^{\infty}\left(  \mathcal{G}\right)  $ be the $\mathbb{R}$-vector space
of functions $f:\mathcal{G}\rightarrow\mathbb{R}$ \ satisfying
\[
\left\Vert f\right\Vert _{\infty}=\sup_{x\in\mathcal{G\smallsetminus A}%
}\left\vert f\left(  x\right)  \right\vert <\infty,
\]
where $\mathcal{A}$ is a subset of $\mathcal{G}$ with measure zero. Let
$L^{1}\left(  \mathcal{G}\right)  $ be the $\mathbb{R}$-vector space of
functions $f:\mathcal{G}\rightarrow\mathbb{R}$ \ satisfying%
\[
\left\Vert f\right\Vert _{1}=%
{\displaystyle\int\limits_{\mathcal{G}}}
\left\vert f\left(  x\right)  \right\vert d\mu<\infty.
\]
For a fixed $w\in L^{1}\left(  \mathcal{G}\right)  $, the mapping%
\[%
\begin{array}
[c]{ccc}%
L^{\infty}\left(  \mathcal{G}\right)  & \rightarrow & L^{\infty}\left(
\mathcal{G}\right) \\
&  & \\
f\left(  x\right)  & \rightarrow & \left(  w\ast f\right)  \left(  x\right)
=\int_{\mathcal{G}}\text{ }w\left(  x-y\right)  f\left(  y\right)  d\mu(y)
\end{array}
\]
is a well-defined \ linear bounded operator satisfying%
\[
\left\Vert w\ast f\right\Vert _{\infty}\leq\left\Vert w\right\Vert
_{1}\left\Vert f\right\Vert _{\infty}.
\]

\begin{remark}
\begin{itemize}
\item[(i)] We recall that $f:\mathbb{R}\rightarrow\mathbb{R}$ is called a
Lipschitz function if there is positive constant $L(f)$ such that $\left\vert
f(x)-f(y)\right\vert \leq L(f)\left\vert x-y\right\vert $ for all $x$ and $y$.

\item[(ii)] Given $\mathcal{X}$, $\mathcal{Y}$, Banach spaces, we denote by
$\mathcal{C}(\mathcal{X},\mathcal{Y})$ the space of continuous functions from
$\mathcal{X}$ to $\mathcal{Y}$.

\item[(iii)] If $\mathcal{Y}=\mathbb{R}$, we use the simplified notation
$\mathcal{C}(\mathcal{X})$.
\end{itemize}
\end{remark}

We fix two bounded Lipschitz functions $S_{E}$, $S_{I}$ satisfying
\[
S_{E}\left(  0\right)  =S_{I}\left(  0\right)  =0.
\]
We also fix $w_{EE}$, $w_{IE}$, $w_{EI}$, $w_{II}\in L^{1}\left(
\mathcal{G}\right)  $, and $h_{E}\left(  x,t\right)  $, $h_{I}\left(
x,t\right)  \in\mathcal{C}(\left[  0,\infty\right]  ,L^{\infty}\left(
\mathcal{G}\right)  )$.

The Wilson-Cowan model on $\mathcal{G}$ is given by the following system of
non-linear integron-differential evolution equations:%
\begin{align*}
\tau\frac{\partial E(x,t)}{\partial t}  &  =-E(x,t)+\\
&  \left(  1-r_{E}E(x,t)\right)  S_{E}\left(  w_{EE}\left(  x\right)  \ast
E(x,t)-w_{EI}\left(  x\right)  \ast I(x,t)+h_{E}\left(  x,t\right)  \right)
\end{align*}%
\begin{align*}
\tau\frac{\partial I(x,t)}{\partial t}  &  =-I(x,t)+\\
&  \left(  1-r_{I}I(x,t)\right)  S_{I}\left(  w_{IE}\left(  x\right)  \ast
E(x,t)-w_{II}(x)\ast I(x,t)+h_{I}\left(  x,t\right)  \right)  ,
\end{align*}
where $\ast$ denotes the convolution in the space variables and $r_{E}$,
$r_{I}\in\mathbb{R}$.

The space $\mathcal{X}:=L^{\infty}\left(  \mathcal{G}\right)  \times
L^{\infty}\left(  \mathcal{G}\right)  $ endowed with the norm
\[
\left\Vert \left(  f_{1},f_{2}\right)  \right\Vert =\max\left\{  \left\Vert
f_{1}\right\Vert _{\infty},\left\Vert f_{2}\right\Vert _{\infty}\right\}
\]
is a real Banach space.

Given $f=\left(  f_{1},f_{2}\right)  \in\mathcal{X}$, and $P(x)$, $Q(x)\in
L^{\infty}\left(  \mathcal{G}\right)  $, we set
\begin{equation}
\boldsymbol{F}_{E}(f)=S_{E}\left(  w_{EE}(x)\ast f_{1}(x)-w_{EI}%
(x)f_{2}(x)+P(x)\right)  , \label{Equation_F_E}%
\end{equation}
and
\begin{equation}
\boldsymbol{F}_{I}(f)=S_{I}\left(  w_{IE}(x)\ast f_{1}(x)-w_{II}(x)\ast
f_{2}(x)+Q(x)\right)  . \label{Equation_F_I}%
\end{equation}
We also set%
\[%
\begin{array}
[c]{ccc}%
\mathcal{X} & \rightarrow & \mathcal{X}\\
&  & \\
f & \rightarrow & \boldsymbol{H}(f),
\end{array}
\]
where $\boldsymbol{H}(f)=\left(  \boldsymbol{H}_{E}(f),\boldsymbol{H}%
_{I}(f)\right)  $ and
\begin{equation}
\boldsymbol{H}_{E}(f)=(1-r_{E}f_{1})\boldsymbol{F}_{E}(f)\text{,
}\boldsymbol{H}_{I}(f)=(1-r_{I}f_{2})\boldsymbol{F}_{I}(f). \label{Equation_H}%
\end{equation}

\begin{remark}
We say that $\boldsymbol{H}$ \ is Lipschitz continuous (or globally Lipschitz)
if there is a constant $L(\boldsymbol{H})$ such that $\left\Vert
\boldsymbol{H}(f)-\boldsymbol{H}(g)\right\Vert \leq L(\boldsymbol{H}%
)\left\Vert f-g\right\Vert $, for all $f$, $g\in\mathcal{X}$. We also say that
$\boldsymbol{H}$ \ is locally Lipschitz continuous (or locally Lipschitz) if
for every $h\in\mathcal{X}$ there exist a ball $B_{R}\left(  h\right)
=\left\{  f\in\mathcal{X};\left\Vert f-h\right\Vert <R\right\}  $ such that
$\left\Vert \boldsymbol{H}(f)-\boldsymbol{H}(g)\right\Vert \leq L(R,
h)\left\Vert f-g\right\Vert $ for all $f$, $g\in B_{R}\left(  h\right)  $.
Since $\mathcal{X}$ is a vector space, without loss of generality, we can
assume that $h=0$.
\end{remark}

\begin{lemma}
\label{Lemma1} With the above notation. If $r_{I}\neq0$ or $r_{E}\neq0$, the
$\boldsymbol{H}:\mathcal{X}\rightarrow\mathcal{X}$ is a well-defined locally
Lipschitz mapping. If $r_{I}=$ $r_{E}=0$, then $\boldsymbol{H}:\mathcal{X}%
\rightarrow\mathcal{X}$ is a well-defined globally Lipschitz mapping.
\end{lemma}

\begin{proof}
We first notice that for $f$, $g\in\mathcal{X}$, by using that $S_{E}$ is
Lipschitz,%
\begin{align*}
&  \left\vert \left(  \boldsymbol{F}_{E}(f)-\boldsymbol{F}_{E}(g)\right)
(x)\right\vert \\
&  \leq L\left(  S_{E}\right)  \left\vert w_{EE}(x)\ast\left(  f_{1}%
(x)-g_{1}\left(  x\right)  \right)  -w_{EI}(x)\left(  f_{2}(x)-g_{2}%
(x)\right)  \right\vert \\
&  \leq L\left(  S_{E}\right)  \left\{  \left\Vert w_{EE}\right\Vert
_{1}\left\Vert f_{1}-g_{1}\right\Vert _{\infty}+\left\Vert w_{EI}\right\Vert
_{1}\left\Vert f_{2}-g_{2}\right\Vert _{\infty}\right\} \\
&  \leq L\left(  S_{E}\right)  \max\left\{  \left\Vert w_{EE}\right\Vert
_{1},\left\Vert w_{EI}\right\Vert _{1}\right\}  \left\Vert f-g\right\Vert ,
\end{align*}
which implies that%
\begin{equation}
\left\Vert \boldsymbol{F}_{E}(f)-\boldsymbol{F}_{E}(g)\right\Vert \leq
L(\boldsymbol{F}_{E})\left\Vert f-g\right\Vert . \label{Estimation 1}%
\end{equation}
Similarly%
\begin{equation}
\left\Vert \boldsymbol{F}_{I}(f)-\boldsymbol{F}_{I}(g)\right\Vert \leq
L(\boldsymbol{F}_{I})\left\Vert f-g\right\Vert , \label{Estimation 2}%
\end{equation}
where $L(\boldsymbol{F}_{I})=L(S_{I})\max\left\{  \left\Vert w_{IE}\right\Vert
_{1},\left\Vert w_{II}\right\Vert _{1}\right\}  $.

Now, by using estimation (\ref{Estimation 1}), \ and the fact that $\left\Vert
\boldsymbol{F}_{E}(f)\right\Vert \leq\left\Vert S_{E}\right\Vert _{\infty}$,
\begin{gather*}
\left\Vert \boldsymbol{H}_{E}(f)-\boldsymbol{H}_{E}(g)\right\Vert =\left\Vert
(1-r_{E}f_{1})\boldsymbol{F}_{E}(f)-(1-r_{E}g_{1})\boldsymbol{F}%
_{E}(g)\right\Vert =\\
\left\Vert \left(  1-r_{E}f_{1}\right)  \left(  \boldsymbol{F}_{E}%
(f)-\boldsymbol{F}_{E}(g)\right)  -r_{E}\boldsymbol{F}_{E}(g)\left(
f_{1}-g_{1}\right)  \right\Vert \\
\leq\left(  1+\left\vert r_{E}\right\vert \left\Vert f_{1}\right\Vert
_{\infty}\right)  \left\Vert \boldsymbol{F}_{E}(f)-\boldsymbol{F}%
_{E}(g)\right\Vert +\left\vert r_{E}\right\vert \left\Vert \boldsymbol{F}%
_{E}(f)\right\Vert \left\Vert f_{1}-g_{1}\right\Vert _{\infty}\\
\leq\left\{  \left(  1+\left\vert r_{E}\right\vert \left\Vert f_{1}\right\Vert
_{\infty}\right)  L(\boldsymbol{F}_{E})+\left\vert r_{E}\right\vert \left\Vert
S_{E}\right\Vert _{\infty}\right\}  \left\Vert f-g\right\Vert .
\end{gather*}
By a similar reasoning using estimation (\ref{Estimation 2}), one gets that%
\[
\left\Vert \boldsymbol{H}_{I}(f)-\boldsymbol{H}_{I}(g)\right\Vert _{\infty
}\leq\left(  \left(  1+\left\vert r_{I}\right\vert \left\Vert f_{2}\right\Vert
_{\infty}\right)  L(\boldsymbol{F}_{I})+\left\vert r_{I}\right\vert \left\Vert
S_{I}\right\Vert _{\infty}\right)  \left\Vert f-g\right\Vert ,
\]
and consequently%
\begin{align}
\left\Vert \boldsymbol{H}(f)-\boldsymbol{H}(g)\right\Vert  &  =\max\left\{
\left\Vert \boldsymbol{H}_{E}(f)-\boldsymbol{H}_{E}(g)\right\Vert _{\infty
},\left\Vert \boldsymbol{H}_{I}(f)-\boldsymbol{H}_{I}(g)\right\Vert _{\infty
}\right\} \nonumber\\
&  \leq\left(  A\left(  1+B\left\Vert f_{}\right\Vert _{\infty}\right)
+C\right)  \left\Vert f-g\right\Vert , \label{Estimation 3}%
\end{align}
where
\[
A:=\max\left\{  L(\boldsymbol{F}_{E}),L(\boldsymbol{F}_{I})\right\}
,B:=\left\{  \left\vert r_{E}\right\vert ,\left\vert r_{I}\right\vert
\right\}  ,C:=\max\left\{  \left\vert r_{E}\right\vert \left\Vert
S_{E}\right\Vert _{\infty},\left\vert r_{I}\right\vert \left\Vert
S_{I}\right\Vert _{\infty}\right\}  .
\]

In the case $r_{E}=r_{I}=0$, estimation (\ref{Estimation 3}) takes the form
\begin{equation}
\left\Vert \boldsymbol{H}(f)-\boldsymbol{H}(g)\right\Vert \leq A\left\Vert
f-g\right\Vert . \label{Estimation 4}%
\end{equation}
Which in turn implies that for $f\in\mathcal{X},$%
\begin{equation}
\left\Vert \boldsymbol{H}(f)\right\Vert \leq\left\Vert \boldsymbol{H}%
(f)-\boldsymbol{H}(0)\right\Vert +\left\Vert \boldsymbol{H}(0)\right\Vert \leq
A\left\Vert f\right\Vert +\left\Vert \left(  \boldsymbol{F}_{E}\left(
0\right)  ,\boldsymbol{F}_{I}\left(  0\right)  \right)  \right\Vert
\label{Estimation 5}%
\end{equation}%
\[
=A\left\Vert f\right\Vert +\left\Vert \left(  S_{E}\left(  0\right)
,S_{I}\left(  0\right)  \right)  \right\Vert \leq A\left\Vert f\right\Vert
+\max\left\{  \left\Vert S_{E}\right\Vert _{\infty},\left\Vert S_{I}%
\right\Vert _{\infty}\right\}  <\infty.
\]
Then estimations (\ref{Estimation 5})- (\ref{Estimation 4})\ imply that
$\boldsymbol{H}$ is a well-defined globally Lipschitz mapping.

We now consider the case $r_{I}\neq0$ or $r_{E}\neq0$. Take $f$, $g\in
B_{R}\left(  0\right)  $, for some $R>0$. Then $\left\Vert f_{1}\right\Vert
_{\infty}<R$, and estimation (\ref{Estimation 3}) takes the form%
\begin{align*}
\left\Vert \boldsymbol{H}(f)-\boldsymbol{H}(g)\right\Vert  &  \leq\left\{
\left(  1+\left\vert r_{E}\right\vert R\right)  L(\boldsymbol{F}%
_{E})+\left\vert r_{E}\right\vert \left\Vert S_{E}\right\Vert _{\infty
}\right\}  \left\Vert f-g\right\Vert \\
&  \leq C\left\Vert f-g\right\Vert \text{, for }f,g\in B_{R}\left(  0\right)
\text{.}%
\end{align*}
Which implies that%
\begin{equation}
\left\Vert \boldsymbol{H}(f)\right\Vert \leq\left\Vert \boldsymbol{H}%
(f)-\boldsymbol{H}(0)\right\Vert +\left\Vert \boldsymbol{H}(0)\right\Vert \leq
C\left\Vert f\right\Vert +\max\left\{  \left\Vert S_{E}\right\Vert _{\infty
},\left\Vert S_{I}\right\Vert _{\infty}\right\}  <\infty. \label{Estimation 6}%
\end{equation}
Then, the restriction of $\boldsymbol{H}$\ to $B_{R}\left(  0\right)  \times
B_{R}\left(  0\right)  $ is a well-defined Lipschitz mapping.
\end{proof}

The estimations given in Lemma \ref{Lemma1} are still valid for functions
depending continuously on a parameter $t$.\ More precisely, take $T>0$ and
$f_{i}\in\mathcal{C}\left(  \left[  0,T\right]  ,\mathcal{U}\right)  $, for
$i=1,2$, where $\mathcal{U}\subset$ $L^{\infty}\left(  \mathcal{G}\right)  $
is an open subset. We assume that
\[
\left(  0,T\right)  \subset f_{i}^{-1}\left(  \mathcal{U}\right)  \text{, for
}i=1,2.
\]
We use the notation $f_{i}=f_{i}\left(  \cdot,t\right)  $, where $t\in\left[
0,T\right]  $, \ $i=1,2$. We replace $P(x)$ by $h_{E}\left(  x,t\right)  $ and
$Q(x)$ by $h_{I}\left(  x,t\right)  $, with $h_{E}\left(  x,t\right)  $,
$h_{I}\left(  x,t\right)  \in\mathcal{C}\left(  \left[  0,\infty\right)
,L^{\infty}\left(  \mathcal{G}\right)  \right)  $. We denote the corresponding
maping $\boldsymbol{H}\left(  f\right)  $ as $\boldsymbol{H}\left(
f,s\right)  $. We also set $\mathcal{X}_{\mathcal{U},T}:=\left[  0,T\right]
\times\mathcal{U}$.

\begin{lemma}
\label{Lemma2}With the above notation, the following assertions hold:

\noindent(i) The mapping $\boldsymbol{H}:\mathcal{X}_{\mathcal{U},T}%
\times\mathcal{X}_{\mathcal{U},T}\rightarrow\mathcal{X}$ is continuous, and
for each $t\in\left(  0,T\right)  $ and each $h\in\mathcal{U}$ there exist
$R>0$ and $L<\infty$ such that%
\[
\left\Vert \boldsymbol{H}\left(  f,s\right)  -\boldsymbol{H}\left(
g,s\right)  \right\Vert \leq L\left\Vert f-g\right\Vert \text{ for }f,g\in
B_{R}(h)\text{, }s\in\left[  0,t\right]  .
\]

\noindent(ii) For $t\in\left(  0,T\right)  $, $f\in\mathcal{U}\times
\mathcal{U}$,
\[%
{\displaystyle\int\nolimits_{0}^{t}}
\left\Vert \boldsymbol{H}\left(  f,s\right)  \right\Vert ds<\infty.
\]

\end{lemma}

\begin{proof}
(i) it follows from Lemma \ref{Lemma1}. By estimations (\ref{Estimation 5})
and (\ref{Estimation 6}), $\left\Vert \boldsymbol{H}\left(  f,s\right)
\right\Vert $ is bounded by a positive constant $C$ depending on $R$, then%
\[%
{\displaystyle\int\nolimits_{0}^{t}}
\left\Vert \boldsymbol{H}\left(  f,s\right)  \right\Vert ds<CT.
\]

\end{proof}

\subsection{The Cauchy problem}

With the above notation, the Cauchy problem for the abstract Wilson-Cowan
system takes the following form%
\begin{equation}
\left\{
\begin{array}
[c]{ll}%
\tau\frac{\partial}{\partial t}\left[
\begin{array}
[c]{c}%
E\left(  x,t\right) \\
I\left(  x,t\right)
\end{array}
\right]  +\left[
\begin{array}
[c]{c}%
E\left(  x,t\right) \\
I\left(  x,t\right)
\end{array}
\right]  =\boldsymbol{H}(\left[
\begin{array}
[c]{c}%
E\left(  x,t\right) \\
I\left(  x,t\right)
\end{array}
\right]  ), & x\in\mathcal{G}\text{, }t\geq0\\
& \\
\left[
\begin{array}
[c]{c}%
E\left(  x,0\right) \\
I\left(  x,0\right)
\end{array}
\right]  =\left[
\begin{array}
[c]{c}%
E_{0}\left(  x\right) \\
I_{0}\left(  x\right)
\end{array}
\right]  \in\mathcal{X}. &
\end{array}
\right.  \label{Cauchy}%
\end{equation}

\begin{theorem}
\label{Theorem1}(i) There exist $T_{0}\in\left(  0,T\right]  $ depending on
$\left[
\begin{array}
[c]{c}%
E_{0}\left(  x\right) \\
I_{0}\left(  x\right)
\end{array}
\right]  \in\mathcal{X}$, such that Cauchy problem (\ref{Cauchy}) has a unique
solution $\left[
\begin{array}
[c]{c}%
E\left(  x,t\right) \\
I\left(  x,t\right)
\end{array}
\right]  $ in $\mathcal{C}^{1}(\left[  0,T_{0}\right)  ,\mathcal{X})$.

\noindent(ii) The solution satisfies%
\begin{gather}
E\left(  x,t\right)  =e^{\frac{-t}{\tau}}E_{0}\left(  x\right)  +%
{\displaystyle\int\nolimits_{0}^{t}}
e^{\frac{-(t-s)}{\tau}}(1-r_{E}E\left(  x,s\right)  )\times\label{Equation_E}%
\\
\left\{  S_{E}\left(  w_{EE}(x)\ast E\left(  x,s\right)  -w_{EI}(x)\ast
I\left(  x,s\right)  +h_{E}\left(  x,s\right)  \right)  \right\}  ds,\nonumber
\end{gather}%
\begin{gather}
I\left(  x,t\right)  =e^{\frac{-t}{\tau}}I_{0}\left(  x\right)  +%
{\displaystyle\int\nolimits_{0}^{t}}
e^{\frac{-(t-s)}{\tau}}(1-r_{E}I\left(  x,s\right)  )\times\label{Equation_I}%
\\
\left\{  S_{I}\left(  w_{IE}(x)\ast E\left(  x,s\right)  -w_{II}(x)\ast
I\left(  x,s\right)  +h_{E}\left(  x,s\right)  \right)  \right\}  ds,\nonumber
\end{gather}
for $t\in\left[  0,T_{0}\right)  $ and $x\in\mathcal{G}$.

\noindent(iii) If $r_{I}=$ $r_{E}=0$, then $T_{0}=\infty$ for any$\left[
\begin{array}
[c]{c}%
E_{0}\left(  x\right) \\
I_{0}\left(  x\right)
\end{array}
\right]  \in\mathcal{X}$, and%
\begin{equation}
\left\vert E\left(  x,t\right)  \right\vert \leq\left\Vert E_{0}\right\Vert
_{\infty}+\tau\left\Vert S_{E}\right\Vert _{\infty}\text{ \ and \ }\left\vert
I\left(  x,t\right)  \right\vert \leq\left\Vert I_{0}\right\Vert _{\infty
}+\tau\left\Vert S_{I}\right\Vert _{\infty}. \label{Eq_Key_bound}%
\end{equation}

\noindent(iv) The solution $\left[
\begin{array}
[c]{c}%
E\left(  x,t\right) \\
I\left(  x,t\right)
\end{array}
\right]  $ in $\mathcal{C}^{1}(\left[  0,T_{0}\right)  ,\mathcal{X})$ depends
continuously on the initial value.
\end{theorem}

\begin{proof}
By Lemma \ref{Lemma2}-(i) and \cite[Lemma 5.2.1 and Theorem 5.1.2]{Milan}, for
each
\[
\left[
\begin{array}
[c]{c}%
E_{0}\left(  x\right) \\
I_{0}\left(  x\right)
\end{array}
\right]  \in\mathcal{X},
\]
there exists a unique $\left[
\begin{array}
[c]{c}%
E\left(  x,t\right) \\
I\left(  x,t\right)
\end{array}
\right]  \in\mathcal{C}(\left[  0,T_{0}\right]  ,\mathcal{X})$ which satisfies
(\ref{Equation_E})-(\ref{Equation_I}). By \ref{Lemma2}-(i) and \cite[Corollary
4.7.5]{Milan}, $\left[
\begin{array}
[c]{c}%
E\left(  x,t\right) \\
I\left(  x,t\right)
\end{array}
\right]  \in\mathcal{C}^{1}(\left[  0,T_{0}\right)  ,\mathcal{X})$ and
satisfies (\ref{Cauchy}). By \cite[Theorem 4.3.4]{C-H}, see also \cite[Theorem
5.2.6]{Milan}, $T_{0}=\infty$ or $T_{0}<\infty$ and $\lim_{t\rightarrow T_{0}%
}\left\Vert \left(  E\left(  t\right)  ,I\left(  t\right)  \right)
\right\Vert =\infty$. In the case $r_{I}=$ $r_{E}=0$, by using that%
\begin{multline*}%
{\displaystyle\int\nolimits_{0}^{t}}
e^{\frac{-(t-s)}{\tau}}\left\vert S_{E}\left(  w_{EE}(x)\ast E\left(
x,s\right)  -w_{EI}(x)\ast I\left(  x,s\right)  +h_{E}\left(  x,s\right)
\right)  \right\vert ds\\
\leq\left\Vert S_{E}\right\Vert _{\infty}%
{\displaystyle\int\nolimits_{0}^{t}}
e^{\frac{-(t-s)}{\tau}}ds<\tau\left\Vert S_{E}\right\Vert _{\infty},
\end{multline*}
and
\begin{multline*}%
{\displaystyle\int\nolimits_{0}^{t}}
e^{\frac{-(t-s)}{\tau}}\left\vert S_{I}\left(  w_{IE}(x)\ast E\left(
x,s\right)  -w_{II}(x)\ast I\left(  x,s\right)  +h_{I}\left(  x,s\right)
\right)  \right\vert ds\\
<\tau\left\Vert S_{I}\right\Vert _{\infty},
\end{multline*}
one shows (\ref{Eq_Key_bound}), which implies that $T_{0}=\infty$.

(iv) It follows from \cite[Lemma 5.2.1 and Theorem 5.2.4]{Milan}.
\end{proof}

\section{Small-world property and Wilson-Cowan models}

\label{Section_2}

After formulating the Wilson-Cowan model on locally compact Abelian groups,
our next step is to find the groups for which the model is compatible with the
description of the cortical networks given by connection matrices. From now
on, we take $r_{I}=$ $r_{E}=0$; in this case, the Wilson-Cowan equations
describe two coupled perceptrons.

\subsection{Compactness and the small-world networks}

The original Wilson-Cowan mo\-del is formulated on $(\mathbb{R},+)$. The
kernels $w_{AB}$, $A$, $B\in\left\{  E,I\right\}  $, which control the
connections between the neurons are supposed to be radial functions of the
form%
\begin{equation}
e^{-C_{AB}\left\vert x-y\right\vert }\text{, or }e^{-D_{AB}\left\vert
x-y\right\vert ^{2}}\text{,} \label{Kernels}%
\end{equation}
where $C_{AB}$, $D_{AB}$\ are positive constants. Since $\mathbb{R}$ is
unbounded, hypothesis (\ref{Kernels}) implies that only short-range
interactions between the neurons occur. The strength of the connections
produced by kernels\ of type (\ref{Kernels}) is negligible outside of a
compact set, then for practical purposes, the interaction between groups of
neurons occur only at small distances.

Nowadays is widely accepted that the brain is a small-world network, see,
e.g., \cite{Sporns}-\cite{Muldoon et al}, and the references therein. The
small-worldness is believed to be a crucial aspect of efficient brain
organization that confers significant advantages in signal processing,
furthermore, the small-world organization is deemed essential for healthy
brain function, see, e.g., \cite{Hilgetag et al},\ and the references therein.
\ A small-world network has a topology that produces short paths across the
whole network, i.e., given two nodes, there is a short path between them (the
six degrees of separation phenomenon). In turn, this implies the existence of
long-range interactions in the network. The compatibility of the
Wilson-Cowan\ model with the small-world network property requires a
non-negligible interaction between any two groups of neurons, i.e.,
$w_{AB}\left(  x\right)  >\varepsilon>0$, for any $x\in\mathcal{G}$, and
for$\ A$, $B\in\left\{  E,I\right\}  $, where the constant $\varepsilon>0$ is
independent of $x$. By Theorem \ref{Theorem1}, it is reasonable to expect that
$w_{AB}$, $A$, $B\in\left\{  E,I\right\}  $ be integrable, then necessarily
$\mathcal{G}$ must be compact.

Finally, we mention that $\left(  \mathbb{R}^{N},+\right)  $ does not have
non-trivial compact subgroups. Indeed, if $x_{0}\neq0$, then $\left\langle
x_{0}\right\rangle =\left\{  nx_{0};n\in\mathbb{Z}\right\}  $ is a non-compact
subgroup of $\left(  \mathbb{R}^{N},+\right)  $, because $\left\{  \left\vert
n\right\vert ;n\in\mathbb{Z}\right\}  $ is not bounded. This last assertion is
equivalent to the Archimedean axiom of the real numbers. In conclusion, the
compatibility between the Wilson-Cowan model and the small-world property
requires changing $\left(  \mathbb{R},+\right)  $ to a compact Abelian group.
The simplest solution is to replace $\left(  \mathbb{R},\left\vert
\cdot\right\vert \right)  $ by a non-Archimedean field $\left(  \mathbb{F}%
,\left\vert \cdot\right\vert _{\mathbb{F}}\right)  $, where the norm satisfies%
\[
\left\vert x+y\right\vert _{\mathbb{F}}\leq\max\left\{  \left\vert
x\right\vert _{\mathbb{F}},\left\vert y\right\vert _{\mathbb{F}}\right\}  .
\]

\subsection{Neurons geometry and discreteness}

Nowadays, there are extensive databases of neuronal wiring diagrams
(connection matrices) of the invertebrates and mammalian cerebral cortex. The
connection matrices are adjacency matrices of weighted directed graphs, where
the vertices represent neurons, regions in a cortex, o neuron populations.
These matrices correspond to the kernels $w_{AB}$, $A$, $B\in\left\{
E,I\right\}  $, then, it seems natural to consider using discrete Wilson-Cowan
models, \cite{Wilson-Cowan 2}, \cite[Chapter 2]{Neural-Fields}. We argue that
two difficulties appear. First, since the connection matrices may be extremely
large, studying the corresponding Wilson-Cowan equations is only possible via
numerical simulations. Second, it seems that the discrete Wilson-Cowan model
is not a good approximation of the continuous Wilson-Cowan model, see
\cite[page 57]{Neural-Fields}. The Wilson-Cowan equations can be formally
discretized by replacing integrals with finite sums. But these discrete models
are relevant only when they are good approximations of the continuous models.
Finally, we want to mention that O. Sporns has proposed the hypothesis that
cortical connections are arranged in hierarchical self-similar patterns,
\cite{Sporns}.

\section{$p$-Adic Wilson-Cowan Models}

The previous section shows that the classical Wilson-Cowan can be formulated
on a large class of topological groups. This formulation does not use any
information about the geometry of the neural interaction, which is encoded in
the geometry of the group $\mathcal{G}$. The next step is to incorporate the
connection matrices into the Wilson-Cowan model, which requires selecting a
specific group. In this section, we propose the $p$-adic Wilson-Cowan models
where $\mathcal{G}$ is the ring of $p$-adic integers $\mathbb{Z}_{p}$.

\subsection{The $p$-adic integers}

This section reviews some basic results on $p$-adic analysis required in this
article. For a detailed exposition on $p$-adic analysis, the reader may
consult \cite{V-V-Z}-\cite{Taibleson}. For a quick review of $p$-adic analysis
the reader may consult \cite{Bocardo-Zuniga-2}.

From now on, $p$ denotes a fixed prime number. The ring of $p-$adic integers
$\mathbb{Z}_{p}$ is defined as the completion of the ring of integers
$\mathbb{Z}$ with respect to the $p-$adic norm $|\cdot|_{p}$, which is defined
as
\begin{equation}
|x|_{p}=%
\begin{cases}
0 & \text{if }x=0\\
p^{-\gamma} & \text{if }x=p^{\gamma}a\in\mathbb{Z},
\end{cases}
\label{p-norm}%
\end{equation}
where $a$ is an integers coprime with $p$. The integer $\gamma=ord_{p}%
(x):=ord(x)$, with $ord(0):=+\infty$, is called the\textit{\ }$p-$\textit{adic
order of} $x$.

Any non-zero $p-$adic integer $x$ has a unique expansion of the form%
\[
x=x_{k}p^{k}+x_{k+1}p^{k+1}+\ldots,\text{ }%
\]
with $x_{k}\neq0$, where $k$ is a non-negative integer, and the $x_{j}$s \ are
numbers from the set $\left\{  0,1,\ldots,p-1\right\}  $. There are natural
field operations, sum, and multiplication, on $p$-adic integers, see, e.g.,
\cite{Koblitz}. The norm (\ref{p-norm}) extends to $\mathbb{Z}_{p}$ as
$\left\vert x\right\vert _{p}=p^{-k}$ for a nonzero $p$-adic integer $x$.

The metric space $\left(  \mathbb{Z}_{p},\left\vert \cdot\right\vert
_{p}\right)  $ is a complete ultrametric space. Ultrametric means that
$\left\vert x+y\right\vert _{p}\leq\max\left\{  \left\vert x\right\vert
_{p},\left\vert y\right\vert _{p}\right\}  $. As a topological space
$\mathbb{Z}_{p}$\ is homeomorphic to a Cantor-like subset of the real line,
see, e.g., \cite{V-V-Z}-\cite{A-K-S}, \cite{Chistyakov}.

For $r\in\mathbb{N}$, denote by $B_{-r}(a)=\{x\in\mathbb{Z}_{p};\left\vert
x-a\right\vert _{p}\leq p^{-r}\}$ \textit{the ball of radius }$p^{-r}$
\textit{with center at} $a\in\mathbb{Z}_{p}$, and take $B_{-r}(0):=B_{-r}$.
{The ball $B_{0}$ equals \textit{the ring of }$p-$\textit{adic integers
}$\mathbb{Z}_{p}$.} We use $\Omega\left(  p^{r}\left\vert x-a\right\vert
_{p}\right)  $ to denote the characteristic function of the ball $B_{-r}(a)$.
Two balls in $\mathbb{Z}_{p}$ are either disjoint or one is contained in the
other. The balls are compact subsets, thus $\left(  \mathbb{Z}_{p},\left\vert
\cdot\right\vert _{p}\right)  $ is a compact topological space.

\subsubsection{Tree-like structures}

The set of $p$-adic integers modulo $p^{l}$, $l\geq1$, consists of all the
integers of the form $i=i_{0}+i_{1}p+\ldots+i_{l-1}p^{l-1}$. These numbers
form a complete set of representatives for the elements of the additive group
$G_{l}=\mathbb{Z}_{p}/p^{l}\mathbb{Z}_{p}$, which is isomorphic to the set of
integers $\mathbb{Z}/p^{l}\mathbb{Z}$ (written in base $p$) modulo $p^{l}$. By
restricting $\left\vert \cdot\right\vert _{p}$ to $G_{l}$, it becomes a normed
space, and $\left\vert G_{l}\right\vert _{p}=\left\{  0,p^{-\left(
l-1\right)  },\cdots,p^{-1},1\right\}  $. With the metric induced by
$\left\vert \cdot\right\vert _{p}$, $G_{l}$ becomes a finite ultrametric
space. In addition, $G_{l}$ can be identified with the set of branches
(vertices at the top level) of a rooted tree with $l+1$ levels and $p^{l}$
branches. By definition, the tree's root is the only vertex at level $0$.
There are exactly $p$ vertices at level $1$, which correspond with the
possible values of the digit $i_{0}$ in the $p$-adic expansion of $i$. Each of
these vertices is connected to the root by a non-directed edge. At level $k$,
with $2\leq k\leq l+1$, there are exactly $p^{k}$ vertices, \ each vertex
corresponds to a truncated expansion of $i$ of the form $i_{0}+\cdots
+i_{k-1}p^{k-1}$. The vertex corresponding to $i_{0}+\cdots+i_{k-1}p^{k-1}$ is
connected to a vertex $i_{0}^{\prime}+\cdots+i_{k-2}^{\prime}p^{k-2}$ at the
level $k-1$ if and only if $\left(  i_{0}+\cdots+i_{k-1}p^{k-1}\right)
-\left(  i_{0}^{\prime}+\cdots+i_{k-2}^{\prime}p^{k-2}\right)  $ is divisible
by $p^{k-1}$. See Figure \ref{Figure 3}. The balls $B_{-r}(a)=a+p^{r}%
\mathbb{Z}_{p}$ are infinite rooted trees.

\begin{figure}[ptb]
\begin{center}
\includegraphics[
height=2.7614in,
width=4.8879in
]{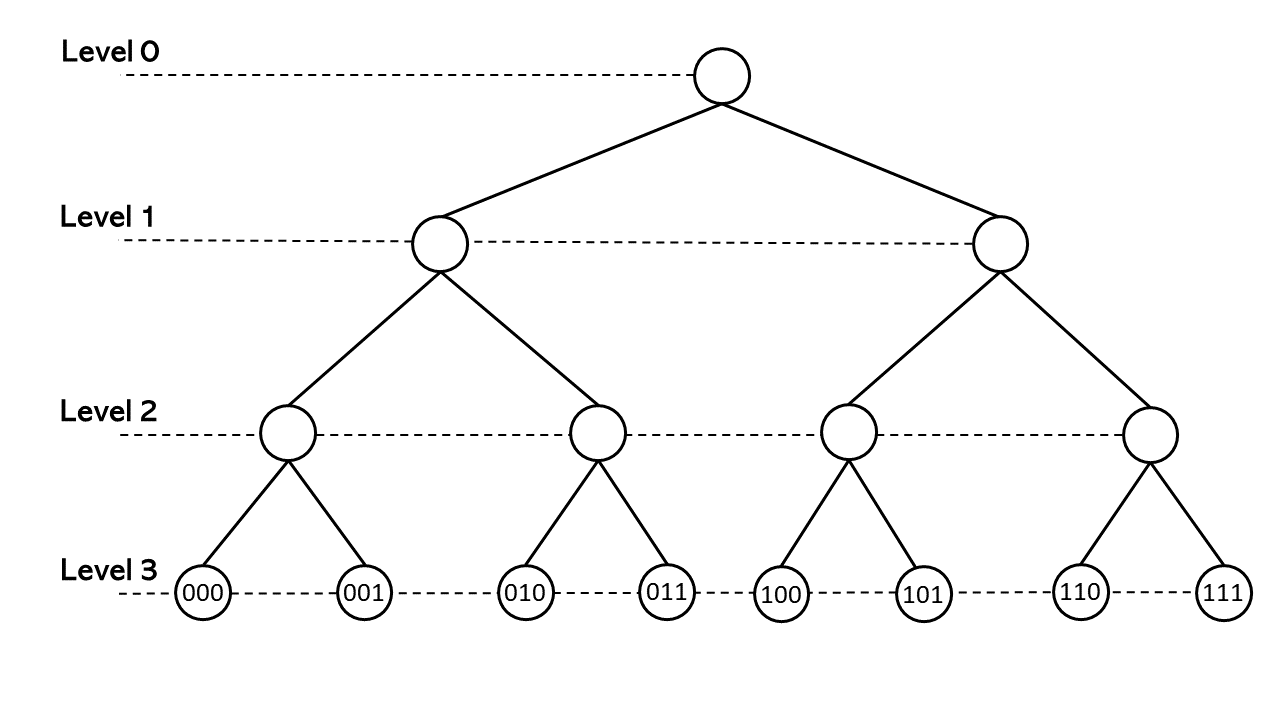}
\end{center}
\caption{The rooted tree associated with the group $\mathbb{Z}_{2}%
/2^{3}\mathbb{Z}_{2}$. The elements of $\mathbb{Z}_{2}/2^{3}\mathbb{Z}_{2}$
have the form $i=i_{0}+i_{1}2+i_{2}2^{2}$,$\;i_{0}$, $i_{1}$, $i_{2}%
\in\{0,1\}$. The distance satisfies $-\log_{2}\left\vert i-j\right\vert _{2}%
=$level of \noindent the first common ancestor of $i$, $j$.}%
\label{Figure 3}%
\end{figure}

\subsection{The Haar measure}

Since $(\mathbb{Z}_{p},+)$ is a compact topological group, there exists a Haar
measure $dx$, which is invariant under translations, i.e., $d(x+a)=dx$,
\cite{Halmos}. If we normalize this measure by the condition $\int
_{\mathbb{Z}_{p}}dx=1$, then $dx$ is unique. It follows immediately that
\[%
{\textstyle\int\limits_{B_{-r}(a)}}
dx=%
{\textstyle\int\limits_{a+p^{r}\mathbb{Z}_{p}}}
dx=p^{-r}%
{\textstyle\int\limits_{\mathbb{Z}_{p}}}
dy=p^{-r}\text{, }r\in\mathbb{N}\text{.}%
\]
In a few occasions, we use the two-dimensional Haar measure $dxdy$ of the
additive group $(\mathbb{Z}_{p}\times\mathbb{Z}_{p},+)$ normalize this measure
by the condition $\int_{\mathbb{Z}_{p}}\int_{\mathbb{Z}_{p}}dxdy=1$. For a a
quick review of the integration in the $p$-adic framework, the reader may
consult \cite{Bocardo-Zuniga-2} and the references therein.

\subsection{The Bruhat-Schwartz space in the unit ball}

A real-valued function $\varphi$ defined on $\mathbb{Z}_{p}$ is \textit{called
Bruhat-Schwartz function (or a test function)} if for any $x\in\mathbb{Z}_{p}$
there exist an integer $l\in\mathbb{N}$ such that%
\begin{equation}
\varphi(x+x^{\prime})=\varphi(x)\text{ for any }x^{\prime}\in B_{l}.
\label{local_constancy}%
\end{equation}
The $\mathbb{R}$-vector space of Bruhat-Schwartz functions supported in the
unit ball is denoted by $\mathcal{D}(\mathbb{Z}_{p})$. For $\varphi
\in\mathcal{D}(\mathbb{Z}_{p})$, the largest number $l=l(\varphi)$ satisfying
(\ref{local_constancy}) is called \textit{the exponent of local constancy (or
the parameter of constancy) of} $\varphi$. A function $\varphi$ in
$\mathcal{D}(\mathbb{Z}_{p})$ can be written as%
\[
\varphi\left(  x\right)  =%
{\displaystyle\sum\limits_{j=1}^{M}}
\varphi\left(  \widetilde{x}_{j}\right)  \Omega\left(  p^{-r_{j}}\left\vert
x-\widetilde{x}_{j}\right\vert _{p}\right)  ,
\]
where the $\widetilde{x}_{j}$, $j=1,\ldots,M$, are points in $\mathbb{Z}_{p}$,
the $r_{j}$, $j=1,\ldots,M$, are non-negative integers, and $\Omega\left(
p^{r_{j}}\left\vert x-\widetilde{x}_{j}\right\vert _{p}\right)  $ denotes the
characteristic function of the ball $B_{-r_{j}}(\widetilde{x}_{j}%
)=\widetilde{x}_{j}+p^{r_{j}}\mathbb{Z}_{p}$.

We denote by $\mathcal{D}^{l}(\mathbb{Z}_{p})$ the $\mathbb{R}$-vector space
of all test functions of the form%
\[
\varphi\left(  x\right)  =%
{\textstyle\sum\limits_{i\in G_{l}}}
\varphi\left(  i\right)  \Omega\left(  p^{l}\left\vert x-i\right\vert
_{p}\right)  \text{, \ }\varphi\left(  i\right)  \in\mathbb{R}\text{,}%
\]
where $i=i_{0}+i_{1}p+\ldots+i_{l-1}p^{l-1}\in G_{l}=\mathbb{Z}_{p}%
/p^{l}\mathbb{Z}_{p}$, $l\geq1$. Notice that $\varphi$ is supported on
$\mathbb{Z}_{p}$ and that $\mathcal{D}(\mathbb{Z}_{p})=\cup_{l\in\mathbb{N}%
}\mathcal{D}^{l}(\mathbb{Z}_{p})$.

The space $\mathcal{D}^{l}(\mathbb{Z}_{p})$ is a finite-dimensional vector
space spanned by the basis $\left\{  \Omega\left(  p^{l}\left\vert
x-i\right\vert _{p}\right)  \right\}  _{i\in G_{l}}$. By identifying
$\varphi\in\mathcal{D}^{l}(\mathbb{Z}_{p})$ with the column vector $\left[
\varphi\left(  i\right)  \right]  _{i\in G_{l}}\in\mathbb{R}^{\#G_{l}}$, we
get that $\mathcal{D}^{l}(\mathbb{Z}_{p})$ is isomorphic to $\mathbb{R}%
^{\#G_{l}}$ endowed with the norm%
\[
\left\Vert \left[  \varphi\left(  i\right)  \right]  _{i\in G_{l}^{N}%
}\right\Vert =\max_{i\in G_{l}}\left\vert \varphi\left(  i\right)  \right\vert
.
\]
Furthermore,
\[
\mathcal{D}^{l}\hookrightarrow\mathcal{D}^{l+1}\hookrightarrow\mathcal{D}%
(\mathbb{Z}_{p}),
\]
where $\hookrightarrow$ denotes a continuous embedding.

\subsection{The $p$-adic version and discrete version of the Wilson-Cowan
models}

The $p$-adic Wilson-Cowan model is obtained by taking $\mathcal{G}%
=\mathbb{Z}_{p}$ and $d\mu=dx$ in (\ref{Cauchy}).

On the other hand, $\left\Vert f\right\Vert _{1}\leq\left\Vert f\right\Vert
_{\infty}$, and
\[
L^{1}(\mathbb{Z}_{p})\supseteq L^{\infty}(\mathbb{Z}_{p})\supseteq
\mathcal{C}(\mathbb{Z}_{p})\supseteq\mathcal{D}(\mathbb{Z}_{p}),
\]
where $\mathcal{C}(\mathbb{Z}_{p})$ denotes the $\mathbb{R}$-space of
continuous functions on $\mathbb{Z}_{p}$ endowed with the norm $\left\Vert
\cdot\right\Vert _{\infty}$. Futhermore, $\mathcal{D}(\mathbb{Z}_{p})$ is
dense in $L^{1}(\mathbb{Z}_{p})$, \cite[Proposition 4.3.3]{A-K-S}, and
consequently, it is also dense in $L^{\infty}(\mathbb{Z}_{p})$ and
$\mathcal{C}(\mathbb{Z}_{p})$.

For the sake of simplicity, we assume that $w_{EE}$, $w_{IE}$, $w_{EI}$,
$w_{II}\in\mathcal{C}(\mathbb{Z}_{p})$, and $h_{E}\left(  x,t\right)  $,
$h_{I}\left(  x,t\right)  \in\mathcal{C}(\left[  0,\infty\right]
,\mathcal{C}(\mathbb{Z}_{p}))$. Theorem \ref{Theorem1} is still valid under
these hypotheses. We use the theory of approximation of evolution equations to
construct good discretizations of the $p$-adic Wilson-Cowan system, see, e.g.,
\cite[Section 5.4]{Milan}.

This theory requires the following hypotheses.

(\textbf{A}). (a) $\mathcal{X}=\left(  \mathcal{C}(\mathbb{Z}_{p}%
)\times\mathcal{C}(\mathbb{Z}_{p})\right)  $, $\mathcal{X}_{l}=\left(
\mathcal{D}^{l}(\mathbb{Z}_{p})\times\mathcal{D}^{l}(\mathbb{Z}_{p})\right)
$, $l\geq1$, endowed with the norm $\left\Vert f\right\Vert =\left\Vert
\left(  f_{1},f_{2}\right)  \right\Vert =\max\left\{  \left\Vert \left(
f_{1}\right)  \right\Vert _{\infty},\left\Vert \left(  f_{2}\right)
\right\Vert _{\infty}\right\}  $ are Banach spaces. It is relevant to mention
that $\mathcal{X}_{l}$ is a subspace of $\mathcal{X}$, and that $\mathcal{X}%
_{l}$ is a subspace of $\mathcal{X}_{l+1}$.

(b) The operator
\[%
\begin{array}
[c]{cccc}%
\boldsymbol{P}_{l}: & \mathcal{X} & \rightarrow & \mathcal{X}_{l}\\
&  &  & \\
& f\left(  x\right)  & \rightarrow & \left(  \boldsymbol{P}_{l}f\right)
\left(  x\right)  =%
{\textstyle\sum\limits_{i\in G_{l}}}
f\left(  i\right)  \Omega\left(  p^{l}\left\vert x-i\right\vert _{p}\right)
\end{array}
\]
is linear and bounded, i.e., $\boldsymbol{P}_{l}\in\mathbb{B}(\mathcal{X}%
,\mathcal{X}_{l})$, and $\left\Vert \boldsymbol{P}_{l}f\right\Vert
\leq\left\Vert f\right\Vert $, for every $f\in\mathcal{X}$.

(c) We set $\boldsymbol{1}_{l}:\mathcal{X}_{l}\rightarrow\mathcal{X}$ to be
the identity operator. Then $\boldsymbol{1}_{l}\in\mathbb{B}(\mathcal{X}%
_{l},\mathcal{X})$, and $\left\Vert \boldsymbol{1}_{l}f\right\Vert =\left\Vert
f\right\Vert $, for every $f\in\mathcal{X}_{l}$.

(d) $\boldsymbol{P}_{l}\boldsymbol{1}_{l}f=f$, for $l\geq1$, $f\in
\mathcal{X}_{l}$.

(\textbf{B}), (\textbf{C}). The Wilson-Cowan system, see (\ref{Cauchy}),
involves the operator $\frac{1}{\tau}\mathbf{1}$, where $\mathbf{1}%
\in\mathbb{B}(\mathcal{X},\mathcal{X})$ is the identity operator. As
approximation, we use $\mathbf{1}\in\mathbb{B}(\mathcal{X}_{l},\mathcal{X}%
_{l})$, for every $l\geq1$. Furthermore,
\[
\lim_{l\rightarrow\infty}\left\Vert \mathbf{P}_{l}f-f\right\Vert =0,
\]
see \cite[Lemma 1]{Zuniga-Nonlinearity}.

(\textbf{D}). For $t\in\left(  0,\infty\right)  $, $\frac{1}{\tau
}\boldsymbol{H}(s,f):\left[  0,t\right]  \times\mathcal{X}\rightarrow
\mathcal{X}$ is a continuous and such that for some $L<\infty$,%
\[
\left\Vert \frac{1}{\tau}\boldsymbol{H}(s,f)-\frac{1}{\tau}\boldsymbol{H}%
(s,g)\right\Vert \leq L\left\Vert f-g\right\Vert \text{,}%
\]
for $0\leq s\leq t$, $f$, $g\in\mathcal{X}$. This assertion is a consequence
of the fact that $\boldsymbol{H}:\mathcal{X}\rightarrow\mathcal{X}$ is a
well-defined globally Lipschitz, see Lemma \ref{Lemma1}.

We use the notation $E\left(  t\right)  =E\left(  \cdot,t\right)  $, $I\left(
t\right)  =I\left(  \cdot,t\right)  \in\mathcal{C}^{1}(\left[  0,T\right)
,\mathcal{X})$, and for the approximations $E_{l}\left(  t\right)
=E_{l}\left(  \cdot,t\right)  $, $I_{l}\left(  t\right)  =I_{l}\left(
\cdot,t\right)  \in\mathcal{C}^{1}(\left[  0,T\right)  ,\mathcal{X})$. The
space discretization of the $p$-adic Wilson-Cowan system (\ref{Cauchy}) is%
\begin{equation}
\left\{
\begin{array}
[c]{l}%
\frac{\partial}{\partial t}\left[
\begin{array}
[c]{c}%
E_{l}\left(  t\right) \\
I_{l}\left(  t\right)
\end{array}
\right]  +\frac{1}{\tau}\left[
\begin{array}
[c]{c}%
E_{l}\left(  t\right) \\
I_{l}\left(  t\right)
\end{array}
\right]  =\frac{1}{\tau}\boldsymbol{P}_{l}\left(  \boldsymbol{H}(\left[
\begin{array}
[c]{c}%
E_{l}\left(  t\right) \\
I_{l}\left(  t\right)
\end{array}
\right]  )\right)  ,\\
\\
\left[
\begin{array}
[c]{c}%
E_{l}\left(  0\right) \\
I_{l}\left(  0\right)
\end{array}
\right]  =\boldsymbol{P}_{l}\left(  \left[
\begin{array}
[c]{c}%
E_{0}\left(  x\right) \\
I_{0}\left(  x\right)
\end{array}
\right]  \right)  \in\mathcal{X}_{l}.
\end{array}
\right.  \label{Cauchy 2}%
\end{equation}
The next step is to obtain an explicit expression for the space discretization
given in (\ref{Cauchy 2}). We need the following formulae.

\begin{remark}
\label{Nota 1}Take
\[
w(x)=%
{\displaystyle\sum\limits_{j\in G_{l}}}
w\left(  j\right)  \Omega\left(  p^{l}\left\vert x-j\right\vert _{p}\right)
\text{, }\phi(y)=%
{\displaystyle\sum\limits_{i\in G_{l}}}
\phi\left(  i\right)  \Omega\left(  p^{l}\left\vert y-i\right\vert
_{p}\right)  \in\mathcal{D}^{l}(\mathbb{Z}_{p}).
\]
Then%
\begin{multline*}
\left(  w\ast\phi\right)  \left(  x\right)  =%
{\displaystyle\int\limits_{\mathbb{Z}_{p}}}
w\left(  x-y\right)  \phi\left(  y\right)  dy=\\%
{\displaystyle\sum\limits_{k\in G_{l}}}
\left\{
{\displaystyle\sum\limits_{i\in G_{l}}}
w\left(  k-i\right)  \phi\left(  i\right)  \right\}  \Omega\left(
p^{l}\left\vert x-k\right\vert _{p}\right)  \in\mathcal{D}^{l}(\mathbb{Z}%
_{p}).
\end{multline*}
Indeed,
\[
\left(  w\ast\phi\right)  \left(  x\right)  =%
{\displaystyle\sum\limits_{j\in G_{l}}}
\text{ }%
{\displaystyle\sum\limits_{i\in G_{l}}}
w\left(  j\right)  \phi\left(  i\right)
{\displaystyle\int\limits_{\mathbb{Z}_{p}}}
\Omega\left(  p^{l}\left\vert x-y-j\right\vert _{p}\right)  \Omega\left(
p^{l}\left\vert y-i\right\vert _{p}\right)  dy.
\]
Changing variables as $z=y-i$, $dz=dx$, in the integral,
\begin{align*}
\left(  w\ast\phi\right)  \left(  x\right)   &  =%
{\displaystyle\sum\limits_{j\in G_{l}}}
\text{ }%
{\displaystyle\sum\limits_{i\in G_{l}}}
w\left(  j\right)  \phi\left(  i\right)
{\displaystyle\int\limits_{\mathbb{Z}_{p}}}
\Omega\left(  p^{l}\left\vert x-z-\left(  i+j\right)  \right\vert _{p}\right)
\Omega\left(  p^{l}\left\vert z\right\vert _{p}\right)  dz\\
&  =%
{\displaystyle\sum\limits_{j\in G_{l}}}
\text{ }%
{\displaystyle\sum\limits_{i\in G_{l}}}
w\left(  j\right)  \phi\left(  i\right)
{\displaystyle\int\limits_{p^{l}\mathbb{Z}_{p}}}
\Omega\left(  p^{l}\left\vert x-z-\left(  i+j\right)  \right\vert _{p}\right)
dz.
\end{align*}
Now, by taking $k=i+j$, and using the fact that $G_{l}$ is an additive group,%
\begin{align*}
\left(  w\ast\phi\right)  \left(  x\right)   &  =%
{\displaystyle\sum\limits_{k\in G_{l}}}
\text{ }%
{\displaystyle\sum\limits_{i\in G_{l}}}
w\left(  k-i\right)  \phi\left(  i\right)
{\displaystyle\int\limits_{p^{l}\mathbb{Z}_{p}}}
\Omega\left(  p^{l}\left\vert x-z-k\right\vert _{p}\right)  dz\\
&  =%
{\displaystyle\sum\limits_{k\in G_{l}}}
\left\{
{\displaystyle\sum\limits_{i\in G_{l}}}
w\left(  k-i\right)  \phi\left(  i\right)  \right\}  \Omega\left(
p^{l}\left\vert x-k\right\vert _{p}\right)  .
\end{align*}

\end{remark}

\begin{remark}
\label{Nota 2}Take $S:\mathbb{R\rightarrow}\mathbb{R}$. Then
\[
S\left(
{\displaystyle\sum\limits_{i\in G_{l}}}
\phi\left(  i\right)  \Omega\left(  p^{l}\left\vert y-i\right\vert
_{p}\right)  \right)  =%
{\displaystyle\sum\limits_{i\in G_{l}}}
S\left(  \phi\left(  i\right)  \right)  \Omega\left(  p^{l}\left\vert
y-i\right\vert _{p}\right)  .
\]
This formula follows from the fact that the supports of the functions
$\Omega\left(  p^{l}\left\vert y-i\right\vert _{p}\right)  $, $i\in G_{l}$,
are disjoint.
\end{remark}

The space discretization of the integro-differential equation in
(\ref{Cauchy 2}) is obtained by computing the term $\boldsymbol{P}_{l}\left(
\boldsymbol{H}(\left[
\begin{array}
[c]{c}%
E_{l}\left(  t\right) \\
I_{l}\left(  t\right)
\end{array}
\right]  )\right)  $ using Remarks \ref{Nota 1}-\ref{Nota 2}. By using the
notation
\[
w_{l}^{AB}=\left[  w_{i}^{AB}\right]  _{i\in G_{l}},\ w_{i}^{AB}=w_{AB}\left(
i\right)  \text{, for }A,B\in\left\{  E,I\right\}  ,
\]%
\[
E_{l}(t)=\left[  E_{i}\left(  t\right)  \right]  _{i\in G_{l}}\text{, }%
E_{i}\left(  t\right)  =E\left(  i,t\right)  \text{, and }I_{l}(t)=\left[
I_{i}\left(  t\right)  \right]  _{i\in G_{l}}\text{, }I_{i}\left(  t\right)
=I\left(  i,t\right)  ,
\]%
\[
h_{l}^{A}\left(  t\right)  =\left[  h_{i}^{A}\left(  t\right)  \right]  _{i\in
G_{l}}\text{, }h_{i}^{A}\left(  t\right)  =h_{A}\left(  i,t\right)  \text{,
for }A\in\left\{  E,I\right\}  ,
\]
and for $\phi_{l}=\left[  \phi_{i}\right]  _{i\in G_{l}}$, $\theta_{l}=\left[
\theta_{i}\right]  _{i\in G_{l}}$,
\[
\phi_{l}\ast\theta_{l}=\left[  \text{ }%
{\displaystyle\sum\limits_{k\in G_{l}}}
\phi_{i-k}\theta_{k}\right]  _{i\in G_{l}}.
\]
With this notation, the announced discretization takes the following form:%
\[
\left\{
\begin{array}
[c]{cc}%
\tau\frac{\partial E_{l}(t)}{\partial t}= & -E_{l}(t)+S_{E}\left(  w_{l}%
^{EE}\ast E_{l}(t)-w_{l}^{EI}\ast I_{l}(t)+h_{l}^{E}\left(  t\right)  \right)
\\
& \\
\tau\frac{\partial I_{l}(t)}{\partial t}= & -I_{l}(t)+S_{I}\left(  w_{l}%
^{IE}\ast E_{l}(t)-w_{l}^{II}(x)\ast I_{l}(t)+h_{l}^{I}\left(  t\right)
\right)  .
\end{array}
\right.
\]

\begin{theorem}
\label{Theorem2}Take $r_{I}=$ $r_{E}=0$, $\left[
\begin{array}
[c]{c}%
E_{0}\left(  x\right) \\
I_{0}\left(  x\right)
\end{array}
\right]  \in\mathcal{X}$, and $T\in\left(  0,\infty\right)  $. Let $\left[
\begin{array}
[c]{c}%
E\left(  t\right) \\
I\left(  t\right)
\end{array}
\right]  \in$ $\mathcal{C}^{1}(\left[  0,T_{0}\right)  ,\mathcal{X})$ be the
solution (\ref{Equation_E})-(\ref{Equation_I}) given in Theorem \ref{Theorem1}%
. Let $\left[
\begin{array}
[c]{c}%
E_{l}\left(  t\right) \\
I_{l}\left(  t\right)
\end{array}
\right]  $\ be the solution of the Cauchy problem (\ref{Cauchy 2}). Then%
\[
\lim_{l\rightarrow\infty}\sup_{0\leq t\leq T}\left\Vert \left[
\begin{array}
[c]{c}%
E_{l}\left(  t\right) \\
I_{l}\left(  t\right)
\end{array}
\right]  -\left[
\begin{array}
[c]{c}%
E\left(  t\right) \\
I\left(  t\right)
\end{array}
\right]  \right\Vert =0.
\]

\end{theorem}

\begin{proof}
We first notice that Theorem \ref{Theorem1} is valid for the Cauchy problem
(\ref{Cauchy 2}), more precisely, this problem has a unique solution $\left[
\begin{array}
[c]{c}%
E_{l}\left(  t\right) \\
I_{l}\left(  t\right)
\end{array}
\right]  $ in $\mathcal{C}^{1}(\left[  0,T_{0}\right)  ,\mathcal{X}_{l})$
satisfying properties akin to the ones stated in Theorem \ref{Theorem1}. Since
$\mathcal{X}_{l}$ is a subspace of $\mathcal{X}$, by applying Theorem
\ref{Theorem1} to the Cauchy problem (\ref{Cauchy 2}), we get the existence of
a unique solution $\left[
\begin{array}
[c]{c}%
E_{l}\left(  t\right) \\
I_{l}\left(  t\right)
\end{array}
\right]  $ in $\mathcal{C}^{1}(\left[  0,T_{0}\right)  ,\mathcal{X})$
satisfying the properties announced in Theorem \ref{Theorem1}. To show that
the solution $\left[
\begin{array}
[c]{c}%
E_{l}\left(  t\right) \\
I_{l}\left(  t\right)
\end{array}
\right]  $ belongs to $\mathcal{C}(\left[  0,T_{0}\right)  ,\mathcal{X}_{l})$,
we use \cite[Theorem 5.2.2]{Milan}. For similar reasoning, the reader may
consult Remark 2, and the proof of Theorem 1 in \cite{ZZ1}. The proof of the
the theorem follows from hypotheses A, B, C, and D by \cite[Theorem
5.4.7]{Milan}. For similar reasoning, the reader may consult the proof of
Theorem 4 in \cite{ZZ1}.
\end{proof}

\section{\label{Section_4}Numerical simulations}

\begin{remark}

We use heat maps to visualize approximations of the solutions of $p$-adic discrete Wilson-Cowan Equations (\ref{Cauchy 2}). The vertical axis gives the position, which is a truncated $p$-adic number. These numbers correspond to a rooted tree's vertices at the top level, i.e., $G_{l}$; see Figure \ref{Figure 3}. By convenience, we include a representation of this tree. The heat maps' colors represent the solutions' values in a particular neuron. For instance, take $p=2$, $l=4$, and 
\begin{equation}\label{Eq:function_heat}
\phi(x)= \Omega(2^{4}\vert x \vert_{2})-\Omega(2^{4}\vert x -2 \vert_{2})+ \Omega(2^{4}\vert x -1\vert_{2})+\Omega(2^{4}\vert x-7 \vert_{2}).
\end{equation}
The corresponding heat map is shown in Figure \ref{Figure0}. 
If the function depends on two variables, say $\phi(x,t)$, where $x\in\mathbb{Z}_{p}$ and $t\in\mathbb{R}$, the corresponding heat map color represents the value of $\phi(x,t)$ at time $t$ and neuron $x$.

\begin{figure}[h]
\begin{center}
\includegraphics[width=0.5\textwidth]{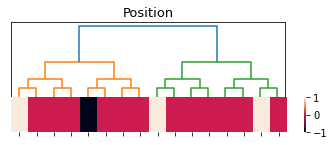}
\end{center}
\caption{Heat map of the function $\phi(x)$, see  (\ref{Eq:function_heat}). Here,  $\phi(0)=\phi(1)=\phi(7)=1$ is white,
$\phi(2)=-1$ is black, and $\phi(x)=0$ is red for $x\neq 0,1,7,2$.}%
\label{Figure0}%
\end{figure}
\end{remark}

We take $\tau=10$, $r_{I}=r_{E}=1$, $p=3$, $l=6$,%
\[
w_{AB}(x)=b_{AB}\exp(\sigma_{AB})-b_{AB}\exp(\sigma_{AB}|x|_{p})\text{, for
}A,B\in\left\{  E,I\right\}  ,
\]
and%
\[
S_{A}(z)=\frac{1}{1+\exp(-v_{A}(z-\theta_{A}))}-\frac{1}{1+\exp(v_{A}%
\theta_{A})}\text{, for }z\in\mathbb{R},\text{ }A\in\left\{  E,I\right\}  .
\]
The kernel $w_{AB}(x)$ is a decresing function of $|x|_{p}$. Thus, near
neurons interact strongly. $S_{A}(z)$ is a sigmoid function satisfying
$S_{A}(0)=0$.

\subsection{Numerical simulation\ 1}

The purpose of this experiment is to show the response of the $p$-adic
Wilson-Cowan network to a short pulse, and a constant stimulus. See Figure
\ref{Figure 4}, \ref{Figure 5}, \ref{Figure 6}. Our results are consistent
with the results obtained by Cowan and Wilson in \cite[Section 2.2.1-Section
2.2.5]{Wilson-Cowan 2}. The pulses are
\begin{equation}
h_{E}(x,t)=3.7\Omega(p^{2}|x-4|_{p})1_{[0,\delta]}(t)\text{, for }%
x\in\mathbb{Z}_{p}\text{, }t\in\left[  0,\delta\right]  \text{,}
\label{Eq_Pulse_1}%
\end{equation}%
\begin{equation}
h_{I}(x,t)=Q\Omega(|x-4|_{p})1_{[0,\delta]}(t)\text{, for }x\in\mathbb{Z}%
_{p}\text{, }t\in\left[  0,\delta\right]  \text{,} \label{Eq_Pulse_2}%
\end{equation}
where $1_{[0,\delta]}(t)$ is the characteristic function of the time interval
$\left[  0,\delta\right]  $, $\delta>0$. We use the following parameters:
$v_{E}=2.75$, $v_{I}=0.3$, $b_{EE}=1.5$, $\sigma_{EE}=4$, $b_{II}=1.8$,
$\sigma_{II}=3$, $\theta_{E}=9$, $\theta_{I}=17$, $b_{IE}=1.35$, $\sigma
_{IE}=6$, $b_{EI}=1.35$, $\sigma_{EI}=6$.

\begin{figure}[h]
\begin{center}
\includegraphics[
height=1.8991in,
width=2.597in
]{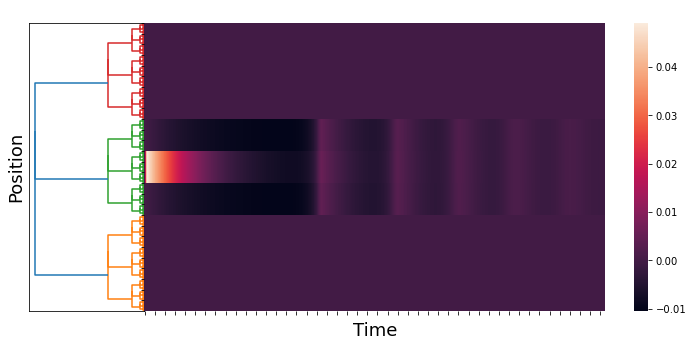}
\end{center}
\caption{An approximation of $E(x,t)$. We take $Q=0$, $\delta=5$. The time
axis goes from $0$ to $100$. The figure shows the response of the network to a brief localized stimulus (the \ pulse given in (\ref{Eq_Pulse_1})). The response is also a pulse. This result is consistent with the numerical results in \cite[Section 2.2.1, Figure 3]{Wilson-Cowan 1}.}%
\label{Figure 4}%
\end{figure}

\begin{figure}[h]
\begin{center}
\includegraphics[
height=1.8991in,
width=2.597in
]{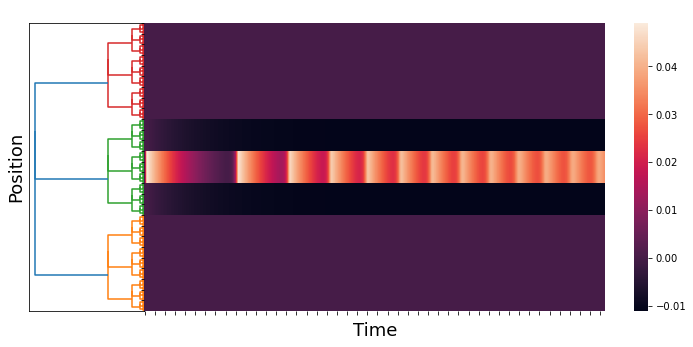}
\end{center}
\caption{An approximation of $E(x,t)$. We take $Q=0$, $\delta=100$. The time
axis goes from $0$ to $100$. The figure shows the response of the network to a
maintained stimulus (see (\ref{Eq_Pulse_1})). The response is a pulse train.
This result is consistent with the numerical results in \cite[Section 2.2.5,
Figure 7]{Wilson-Cowan 1}.}%
\label{Figure 5}%
\end{figure}

\begin{figure}[h]
\begin{center}
\includegraphics[
height=1.8991in,
width=2.5417in
]{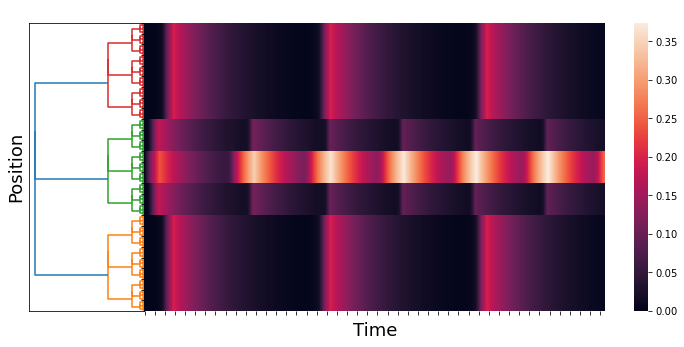}
\end{center}
\caption{An approximation of $E(x,t)$. We take $Q=-30$, $\delta=100$. The time
axis goes from $0$ to $100$. The figure shows the response of the network to a
maintained stimulus (see (\ref{Eq_Pulse_1})-(\ref{Eq_Pulse_2})). The response
is a pulse train in space and time. This result is consistent with the
numerical results in \cite[Section 2.2.7, Figure 9]{Wilson-Cowan 1}.}%
\label{Figure 6}%
\end{figure}

\newpage
\subsection{Numerical simulation \ 2}
In \cite[Section 3.3.1]{Wilson-Cowan 2}, Wilson and Cowan applied their model to the spatial hysteresis in the one-dimensional tissue model. In this experiment, a human subject was exposed to a binocular stimulus. The authors
used sharply peaked Gaussian distributions to model the stimuli. The two stimuli were symmetrically moved apart by a small increment and re-summed, and the network response was allowed to reach equilibrium.

Initially, the two peaks (stimuli) were very closed; the network response
consisted of a single pulse (peak), see \cite[Section 3.3.1, Figure
13-A]{Wilson-Cowan 2}. Then, the peaks separated from each other (i.e., the
disparity between the two stimuli increased). The network response was a pulse
in the middle of the binocular stimulus until a critical disparity was
reached. At this stimulus disparity, the single pulse (peak) decayed rapidly
to zero, and twin response pulses formed at the locations of the now rather
widely separated stimuli, see \cite[Section 3.3.1, Figure 13-B]{Wilson-Cowan
2}.

Following this, the stimuli were gradually moved together again in the same
form until they essentially consisted of one peak. But the network response
consisted of two pulses, see \cite[Section 3.3.1, Figure 13-C]{Wilson-Cowan 2}.

The classical Wilson-Cowan model and our $p$-adic version can predict the
results of this experiment. We use the function%
\begin{equation}
\widetilde{h}_{E}(x,t)=e^{-(30(0.5-\mathfrak{m}(x))-0.5t)^{2}}%
+e^{-(30(0.5-\mathfrak{m}(x))+0.5t)^{2}} \label{Stimuli_1}%
\end{equation}
to model the stimuli in the case where the peaks do not move together, and%
\begin{equation}
h_{E}(x,t)=\widetilde{h}_{E}(x,t)1_{[0,18]}(t)+\widetilde{h}_{E}%
(x,36-t)1_{[18,36]}(t) \label{Stimuli_2}%
\end{equation}
to model the stimuli in the case where the peaks gradually move together. The
function $\mathfrak{m}:\mathbb{Z}_{p}\rightarrow\mathbb{R}$ is the Monna map,
see \cite{Monna}.

Figure \ref{Figure 7} shows the stimuli (see (\ref{Stimuli_1})) and the
network response when the stimuli peaks are gradually separated. The network
response begins with a single pulse. When a critical disparity threshold is
reached, the response becomes a twin pulse. Which is the prediction of the
classical Wilson-Cowan model, see \cite[Section 3.3.1, Figures 13-A,
13-B.]{Wilson-Cowan 2}.

Figure \ref{Figure 8} depicts the stimuli and the network response in the
instance where the stimuli peaks gradually split, and finally move together.
The network response at the end of the experiment consists of twin pulses.
This finding is consistent with that of the classical Wilson-Cowan model,
\cite[Section 3.3.1, Figure 13-C]{Wilson-Cowan 2}.

\begin{figure}[h]
\begin{center}
\includegraphics[width=0.7\textwidth]{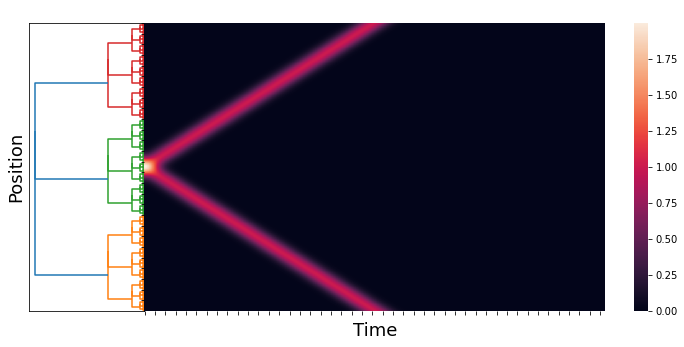}
\includegraphics[width=0.7\textwidth]{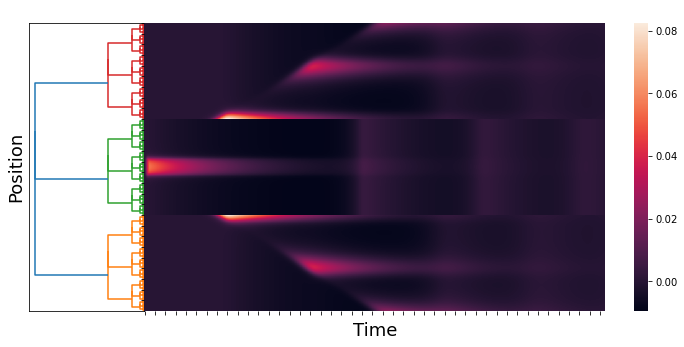}
\end{center}
\caption{An approximation of $\widetilde{h}_{E}(x,t)$ and $E(x,t)$. We take
$h_{I}(x,t)\equiv0$, $p=3$, $l=6$, the kernels $w_{AB}$ are as in the
Simulation 1, and the $h_{E}(x,t)$ as in (\ref{Stimuli_1}). The time axis goes
from $0$ to $60$. The first figure is the stimuli and the second figure is the
response of the network.}
\label{Figure 7}%
\end{figure}

\begin{figure}[h]
\begin{center}
\includegraphics[width=0.7\textwidth]{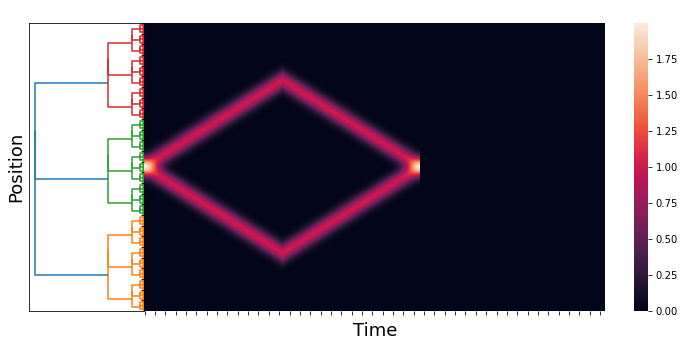}
\includegraphics[width=0.7\textwidth]{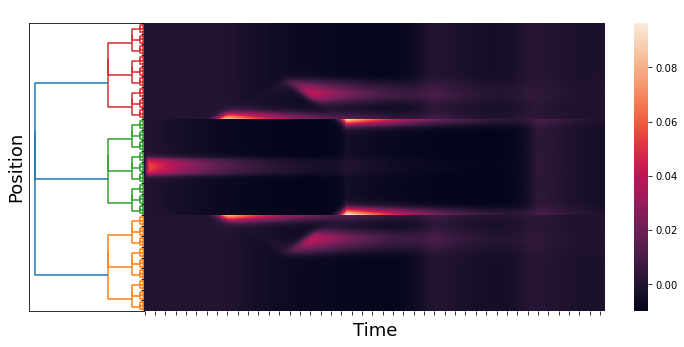}
\end{center}
\caption{An approximation of $h_{E}(x,t)$ and $E(x,t)$. We take $h_{I}(x,t)\equiv0$, $p=3$, $l=6$, the kernels $w_{AB}$ are as in the Simulation 1, and the $h_{E}(x,t)$ as in (\ref{Stimuli_2}). The time axis goes from $0$ to $60$. The first figure is the stimuli and the second figure is the response of the network.}
\label{Figure 8}%
\end{figure}

\bigskip
\section{\label{Section_5}$p$-Adic kernels and connection matrices}

There have been significant theoretical and experimental developments in
comprehending the wiring diagrams (connection matrices) of the cerebral cortex
of invertebrates and mammals over the last thirty years, see, for example,
\cite{Sporns et al 1}-\cite{Skoch et al} and the references therein. The
topology of cortical neural networks is described by connection matrices.
Building dynamic models from experimental data recorded in connection matrices
is a very relevant problem.

We argue that our $p$-adic Wilson-Cowan model provides meaningful dynamics on
networks whose topology comes from a connection matrix. Figure \ref{Figure 9}
depicts the connection matrix of the cat cortex, see, e.g., \cite{Scanell et
al 1}-\cite{Scanell et al 2}, and the matrix of the kernel $w_{EE}$ used in
the Simulation 1. The $p$-adic methods are relevant only if the connection
matrices can be very well approximated for matrices coming from
discretizations of $p$-adic kernels. This is an open problem. Here, we show
that such an approximation is feasible for the cat cortex connection matrix.

\begin{figure}[h]
\begin{center}
\includegraphics[width=0.45\textwidth]{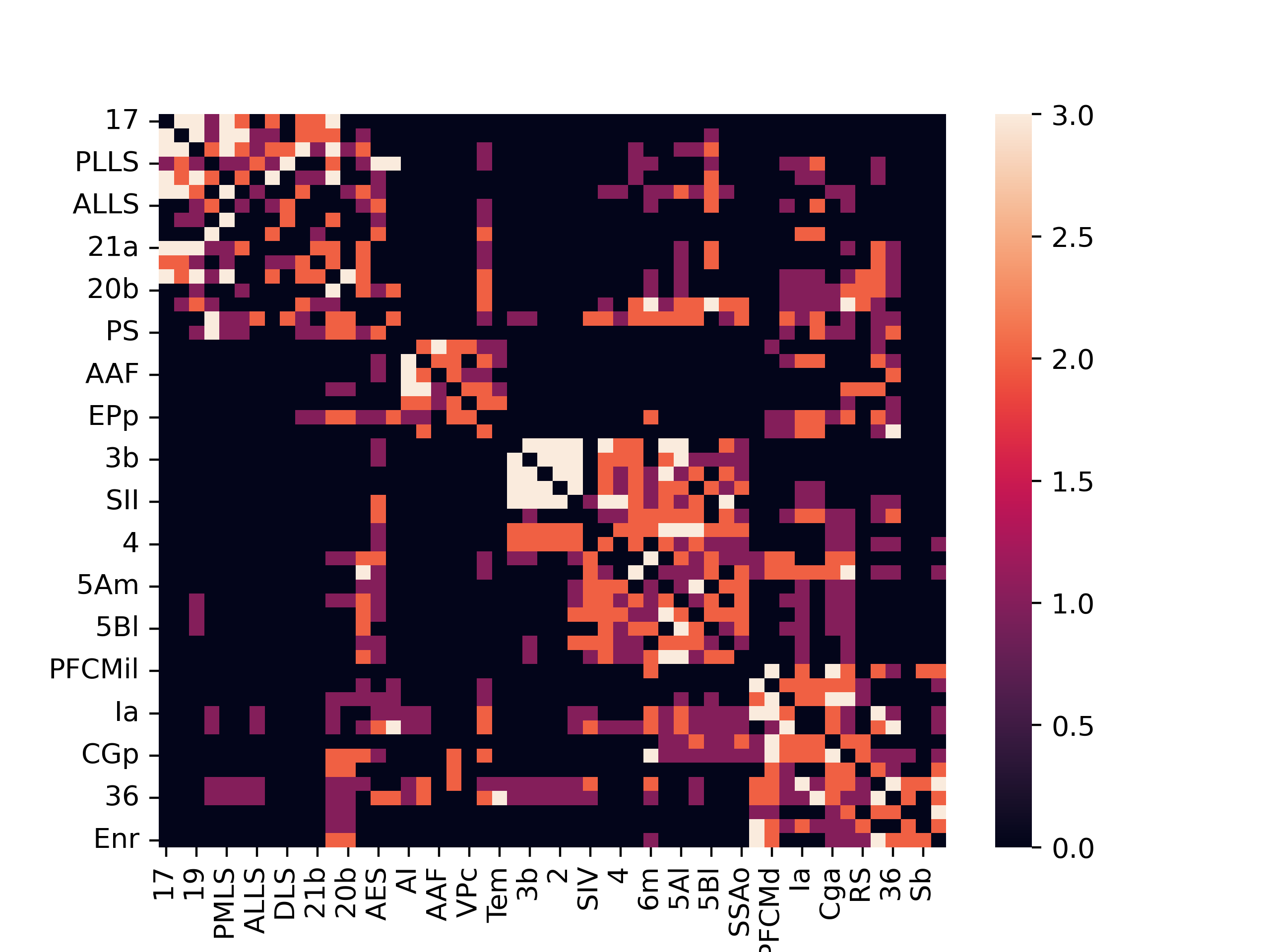}
\includegraphics[width=0.4\textwidth]{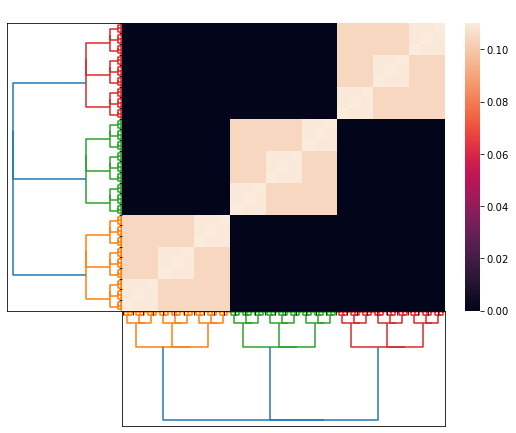}
\end{center}
\caption{The left matrix is the connection matrix of the cat cortex. The right matrix corresponds to a discretization of the kernel $w_{EE}$ use in the simulation 1.}%
\label{Figure 9}%
\end{figure}

\begin{figure}[h]
\begin{center}
\includegraphics[width=0.31\textwidth]{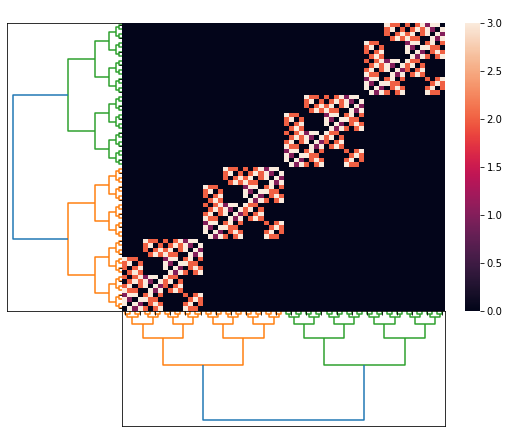}
\includegraphics[width=0.31\textwidth]{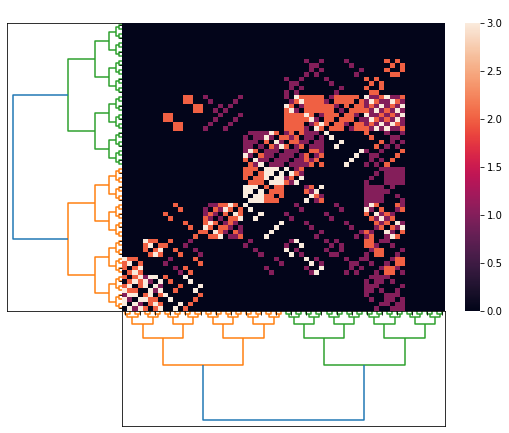}
\includegraphics[width=0.31\textwidth]{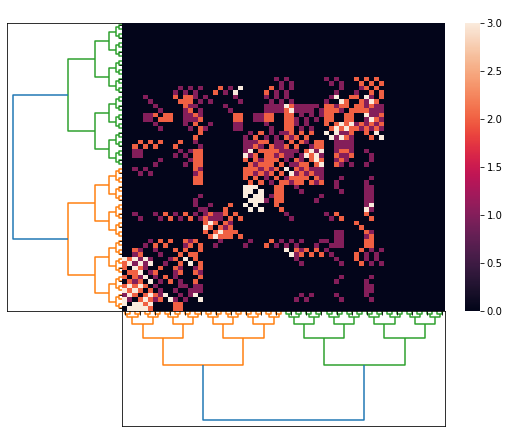}
\end{center}
\caption{Three $p$-adic approximations for the connection matrix of the cat
cortex. We take $p=2$, $l=6$. The first approximation uses $r=0$, the second
$r=3$, and the last $r=5$.}%
\label{Figure 10}%
\end{figure}

By using the above procedure, we replace the excitatory-excitatory relation
term $w_{EE}\ast E$ by $\int_{\mathbb{Z}_{p}}K_{r}(x,y)E(y)dy$ but keep the
other kernels as in Simulation 1. For the stimuli we use $h_{E}=3.5\Omega
(p^{2}|x-1|_{p})$, with $p=2$, $l=6$, and $h_{I}(x)=-30$. In the Figure
\ref{Figure 10}, we show three different approximations for the cat cortex
connection matrix using $p$-adic kernels. The black area in the right matrix
in Figure \ref{Figure 10} (which correspond to zero entries) comes from the
process of adjusting the size of origin matrix to $2^{6}\times2^{6}$.

The corresponding $p$-adic network responses are shown in Figure
\ref{Figure 11} for different values of $r$. In the case, $r=0$, the
interaction between the neurons is short-range, while in the case $r=5$, there
is long-range interaction. The response in the case $r=0$ is similar to the
one presented in Simulation 1; see Figure \ref{Figure 6}. When the connection
matrix gets close to the cat cortex matrix (see Figure \ref{Figure 10}), which
is when the matrix allows more long-range connections, the response of the
network presents more complex patterns (see Figure \ref{Figure 11}).

\begin{figure}[h]
\begin{center}
\includegraphics[width=0.4\textwidth]{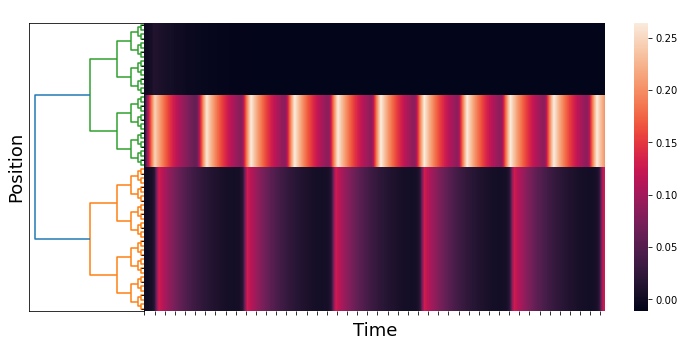}
\includegraphics[width=0.4\textwidth]{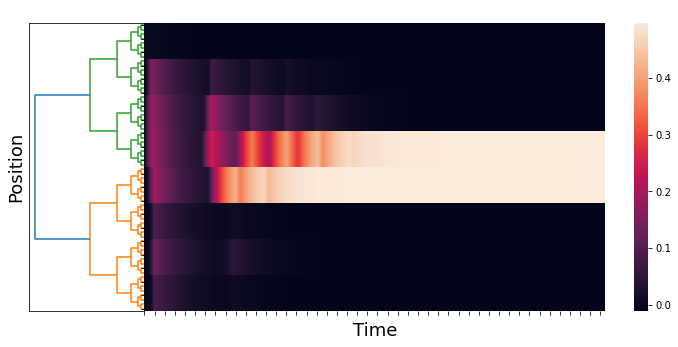}
\includegraphics[width=0.4\textwidth]{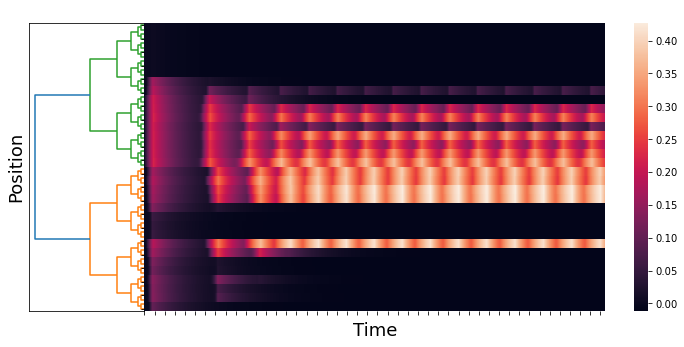}
\end{center}
\caption{We use $p=2$, $l=6$, and time axis goes from $0$ to $150$. The left
image uses $r=0$, the right one uses $r=3$, and the central one uses $r=5$.}%
\label{Figure 11}%
\end{figure}

\section{\label{Section_6}Final\ Discussion}

The Wilson--Cowan model describes interactions between populations of  excitatory and inhibitory neurons. This model constitutes a relevant mathematical tool for understanding cortical tissue's functionality.
On the other hand, in the last twenty-five years, there has been tremendous experimental development in understanding the cerebral cortex's neuronal wiring in invertebrates and mammalians. Employing different experimental techniques, the wiring patterns can be described by connection matrices. A such matrix is just an adjacency matrix of a directed graph whose nodes represent neurons, groups of neurons, or  portions of the cerebral cortex. The oriented edges represent the strength of the connections between two groups of neurons.
This work explores the interplay between the classical Wilson-Cowan model and the connections matrices.

Nowadays, it is widely accepted that the networks in the cerebral cortex of mammalians have the small-world property, which means a non-negligible interaction exists between any two groups of neurons in the network. The classical Wilson-Cowan model is not compatible with the small-world property.
We show that the original Wilson-Cowan model can be formulated on any topological group, and the Cauchy problem for the underlying equations of the model is well-posed. We gave an argument showing that the small-world property requires that the group be compact, and consequently, the classical model should be discarded. In practical terms, the classical Wilson-Cowan model cannot incorporate the experimental information contained in the connection matrices. We proposed a $p$-adic Wilson-Cowan model, where the neurons are organized in an infinite rooted tree. We present numerical experiments showing that this model can explain several phenomena like the classical model. The new model can incorporate the experimental information coming from the connection matrices.

\end{document}